\documentclass{sig-alt-full}

\usepackage{amssymb,amsmath}
\usepackage{theorem}
\usepackage{graphics}
\usepackage{amsfonts}
\usepackage{mathrsfs}
\usepackage{amscd}
\usepackage{color}
\usepackage{url}
\usepackage[plainpages=false,pdfpagelabels,colorlinks=true,citecolor=blue,hypertexnames=false]{hyperref}

\newcommand{\bigO}{{\mathcal{O}}}
\newcommand{\bigOsoft}{\tilde{\mathcal{O}}}
\newcommand{\lc}{\operatorname{lc}}

\newcommand{\bideg}{\operatorname{bideg}}
\newcommand{\HO}{\mathcal H}
\newcommand{\bN}{\mathbb{N}}
\newcommand{\cM}{\mathcal M}

\def\gathen#1{{#1}}

\newtheorem{lemma}{Lemma}
\newtheorem{theorem}[lemma]{Theorem}
\newtheorem{fact}[lemma]{Fact}
\newtheorem{cor}[lemma]{Corollary}

\newtheorem{problem}{Problem}

\newtheorem{definition}{Definition} 

\begin{document}
\conferenceinfo{ISSAC XXXX,}{}
\CopyrightYear{2010}
\crdata{}

\title{Complexity of Creative Telescoping\\ for Bivariate Rational Functions%
\titlenote{\small %
We warmly thank the referees for their very helpful comments.
---
AB and FC were supported in part
by the Microsoft Research\,--\,Inria Joint Centre,
and SC and ZL by a grant of the
National Natural Science Foundation of China (No.~60821002).%
\vspace{-25pt}}}
\newfont{\authfntsmall}{phvr at 11pt}
\newfont{\eaddfntsmall}{phvr at 9pt}
\def\more-auths{%
\end{tabular}
\begin{tabular}{c}}
\numberofauthors{2}
\author{
\alignauthor {\authfntsmall Alin Bostan, Shaoshi Chen, Fr\'ed\'eric Chyzak}\\
\affaddr{Algorithms Project-Team, INRIA Paris-Rocquencourt}\\
\affaddr{78153 Le Chesnay (France)}\\
\email{\eaddfntsmall{\{alin.bostan,shaoshi.chen,frederic.chyzak\}@inria.fr}}
\alignauthor {\authfntsmall Ziming Li}\\
\affaddr{Key Laboratory of Mathematics Mechanization, Academy of Mathematics and System Sciences}\\
\affaddr{100190 Beijing (China)}\\
\email{\eaddfntsmall{zmli@mmrc.iss.ac.cn}}}

\maketitle
\begin{abstract}
The long-term goal initiated in this work is to obtain fast
algorithms and implementations for definite integration in
Almkvist and Zeilberger's framework of (differential)
creative telescoping.
Our complexity-driven approach is to
obtain tight degree bounds on the various expressions
involved in the method. To make the
problem more tractable, we restrict
to \emph{bivariate rational\/} functions. By considering this
constrained class of inputs, we are able to blend the general method
of creative telescoping with the well-known Hermite reduction. We
then use our new method to compute diagonals of rational power series
arising from combinatorics.
\end{abstract}

\vspace{1mm}
\noindent
{\bf Categories and Subject Descriptors:} \\
\noindent I.1.2 [{\bf Computing Methodologies}]: Symbolic and
Algebraic Manipulations --- \emph{Algebraic Algorithms}

\vspace{1mm}
\noindent {\bf General Terms:} Algorithms, Theory.

\vspace{1mm}

\begin{section}{Introduction}
The long-term goal of the research initiated in the present work is to
obtain fast algorithms and implementations for the definite integration of
general special functions, in a complexity-driven perspective.

As most special-function integrals cannot be expressed in closed form,
their evaluation cannot be based on table look-ups only, and even when
closed forms are available, they may prove to be intractable in
further manipulations.  In both cases, the difficulty can be mitigated
by representing functions by annihilating differential operators.
This motivated Zeilberger to introduce a method now known as
\emph{creative telescoping\/}~\cite{Zeilberger1990}, which
applies to a large class of special functions:
the D-finite functions~\cite{Lipshitz1989} defined by
sets of linear differential equations of any order, with polynomial
coefficients.  Zeilberger's method applies in general to multiple
integrals and sums.

\begin{figure*} \begin{center} \renewcommand{\arraystretch}{1.2}
\tabcolsep4pt
\begin{tabular}{l|l|ll|ll|ll}
  \hline
  & Method & $\deg_{D_x}(L)$ & $\deg_x(L)$ & $\deg_x(g)$ & $\deg_y(g)$ & Complexity & \\ \hline
  Minimal & Hermite reduction (new) & $\leq d_y$ & $\bigO(d_xd_y^2)$ & $\bigO(d_xd_y^2)$ &  $\bigO(d_y^2)$ & $\bigOsoft(d_xd_y^{\omega+3})$ & Las Vegas \\ \cline{2-8}
  Telescoper & Almkvist and Zeilberger & $\leq d_y$ & $\bigO(d_xd_y^2)$ & $\bigO(d_xd_y^2)$ & $\bigO(d_y^2)$ & $\bigOsoft(d_xd_y^{2\omega+2})$ & Las Vegas \\ \hline
  Nonminimal& Lipshitz elimination & $\leq6(d_x+1)(d_y+1)$ & $\bigO(d_x d_y)$ & $\bigO(d_x^2d_y)$ & $\bigO(d_xd_y^2)$ & $\bigO(d_x^{3\omega} d_y ^{3\omega})$ & deterministic \\\cline{2-8}
  Telescoper & Cubic size & $\leq6d_y$ & $\bigO(d_xd_y)$ & $\bigO(d_xd_y)$ & $\bigO(d_y^2)$ & $\bigO(d_x^{\omega}d_y^{3\omega})$ & deterministic \\\hline
\end{tabular}
\caption{Complexity of creative telescoping methods (under Hyp.~(H')), together with bounds on output}\label{fig:complexity}
\end{center}
\vskip-15pt
\end{figure*}


A sketch of Zeilberger's method is as follows.
Given a D-finite function~$f$ of the variables
$x$ and~$y$, the definite integral
$ F(x)= \int_\alpha^\beta f(x, y)\,dy$
is D-finite, and a linear differential equation satisfied by~$F$ can
be constructed~\cite{Zeilberger1990}.
To explain this, let $k$ be
a field of characteristic zero, $D_x$ and~$D_y$ be the usual
derivations on the rational-function field $k(x, y)$,
both restricting to zero on~$k$,
and let $k(x, y)\langle D_x, D_y \rangle$ be the ring of linear differential
operators over $k(x, y)$.
The heart of the method is to solve the
\emph{differential telescoping equation}
\eqref{eq:CT} below
for $L\in k[x]\langle D_x \rangle\setminus \{0\}$ and $g=R(f)$ for
some $R\in k(x, y)\langle D_x, D_y\rangle$.
The operator~$L$ is called a
\emph{telescoper\/} for~$f$, and $g$~a \emph{certificate\/} of~$L$ for~$f$.
Under the assumption
\begin{equation*}
\lim_{y\rightarrow\alpha}g(x,y)=\lim_{y\rightarrow\beta}g(x,y)
\quad\text{for $x$ in some domain},
\vspace{-0.1cm}
\end{equation*}
$L(x, D_x)$ is then proved to be an annihilator of $F$.

The main emphasis in works since the 1990's has been on finding
telescopers of order minimal over all telescopers for~$f$, which are
called \emph{minimal telescopers}.
(Two minimal telescopers differ by a multiplicative factor in~$k(x)$.)
In view of the computational
difficulty of solving~\eqref{eq:CT}, there has been special
attention to subclasses of inputs. Of particular importance is the
case of hyperexponential functions, defined by first-order
differential equations, studied by Almkvist and Zeilberger
in~\cite{Almkvist1990}. Their method is a direct differential
analogue of Zeilberger's algorithm for the recurrence
case~\cite{Zeilberger1991}.

On the other hand, very little is known about the complexity of
creative telescoping: the only related result seems to be an
analysis in~\cite{Gerhard2004} of an algorithm for hyperexponential
indefinite integration.
In order to get complexity estimates, we simplify the
problem by restricting to a smaller class of inputs, namely that of
bivariate rational functions.
Although restricted, this class already has many applications,
for instance in combinatorics,
where many nontrivial problems are encoded as diagonals of rational
formal power series, themselves expressible as integrals.
Our goal thus reads as follows.

\begin{problem}
Given $f=P/Q\in k(x, y)\setminus \{0\}$, find a pair $(L, g)$ with
$L = \sum_{i=0}^\rho \eta_i(x)D_x^i$ in $k[x]\langle D_x \rangle
\setminus \{0\}$ and $g$ in $k(x, y)$ such that
\begin{equation}\label{eq:CT}
L(x, D_x)(f) = D_y(g).
\end{equation}
\end{problem}

By considering this more constrained class of inputs, we are indeed
able to blend the general method of creative telescoping with the
well-known Hermite reduction~\cite{Hermite1872}.

Essentially two algorithms for minimal telescopers can be
found in the literature:
The classical way~\cite{Almkvist1990}
is to apply a differential analogue of Gosper's indefinite summation
algorithm, which reduces the problem to solving an auxiliary linear
differential equation for polynomial solutions.
An algorithm developed later in~\cite{GeddesLe2002} (see
also~\cite{Le2000}) performs Hermite reduction on~$f$ to get an
additive decomposition of the form $f = D_y(a) + \sum_{i=1}^m
u_i/v_i$, where~the $u_i$ and $v_i$ are in $k(x)[y]$ and the~$v_i$
are squarefree. Then, the algorithm in~\cite{Almkvist1990} is
applied to each~$u_i/v_i$ to get a telescoper~$L_i$ minimal for it.
The least common left multiple of the $L_i$'s is then proved to be a
minimal telescoper for~$f$. This algorithm performs well only for
specific inputs (both in practice and from the complexity
viewpoint), but it inspired our Lemma~\ref{le:dGf-for-F'/F}
via~\cite{Le2000}.

As a first contribution in this article, we present a new, provably
faster algorithm for computing minimal telescopers for bivariate
rational functions.
Instead of a single use of Hermite reduction as in~\cite{Le2000}, we apply
Hermite reduction to the $D_x^i(f)$'s, iteratively for~$i=0,1,\dots$,
which yields
\begin{equation} \label{EQ:incremental}
 D_x^i (f)  =  D_y (g_i) + \frac{w_i}{w}
\end{equation}
for some factor~$w$ of the squarefree part of the denominator of~$f$.
If $\eta_0, \ldots, \eta_\rho \in k(x)$ are not all zero and such that
$\sum_{i=0}^\rho \eta_i w_i = 0$, then the operator~$\sum_{i=0}^\rho
\eta_i D_x^i$ is a telescoper for~$f$, and more specifically, the first nontrivial linear
relation obtained in this way yields a minimal telescoper for~$f$.

As a second contribution, we give the first proof of a polynomial
complexity for creative telescoping on a specific class of inputs,
namely on bivariate rational functions.
For \emph{minimal\/} telescopers, only a polynomial bound on~$d_x$ (but none
on~$d_y$) was given for special inputs in~\cite{GeddesLe2002};
more specifically, we derive complexity estimates for all
mentioned methods (see Fig.~\ref{fig:complexity}), showing that
our approach is faster.
Furthermore, we analyse the bidegrees of \emph{non minimal\/}
telescopers generated by other approaches:
Lipshitz' work~\cite{Lipshitz1988} can be
rephrased into an existence theorem for telescopers with
polynomial size; the approach followed in the recent work on
algebraic functions~\cite{BCLSS2007} leads to
telescopers of smaller degree sizes.
These are new instances of the philosophy, promoted in~\cite{BCLSS2007},
that relaxing minimality
can produce smaller outputs.

A third contribution is a fast Maple implementation~\cite{OurSoft},
incorporating a
careful implementation of the original Hermite reduction algorithm,
making use of the special form of~$w_i/w$ in~\eqref{EQ:incremental}
and of usual modular techniques (probabilistic rank estimate) to
determine when to invoke the solver for linear algebraic equations.
Experimental results indicate that our implementation
outperforms Maple's core routine.

Note that for the fastest method we propose, denoted by~\verb+H1+ in
Tables~\ref{tab:random}--\ref{tab:algos}, we chose to output the
certificate as a mere sum of (small) rational functions, without any
form of normalisation. This choice seems to be uncommon for
creative-telescoping algorithms, but a motivation is how the
certificate is used in practice: Very often, like for applications
to diagonals in \S\,\ref{sec:implementation}, the certificate is
actually not needed. In other applications, the next step of the
method of creative telescoping is to integrate~\eqref{eq:CT} between
$\alpha$ and~$\beta$, leading to $L(F)(x) = g(x, \alpha) - g(x,
\beta)$. Therefore, only evaluations of the certificate are really
needed, and normalisation can be postponed to after specialising at
$\alpha$ and~$\beta$.

The end of this section, \S\,\ref{sec:background}, provides classical
complexity results, notation, and hypotheses that will be used
throughout.  We then study Hermite reduction over~$k(x)$ in
\S\,\ref{sec:Hermite-reduction}, proving output degree bounds and a
low-complexity algorithm.  This is then applied in
\S\,\ref{sec:minimal-order} to derive our new algorithm for creative
telescoping, and to compare its complexity with that of Almkvist and
Zeilberger's approach.
For nonminimal telescopers, we show the existence
of some of lower arithmetic size in \S\,\ref{sec:nonminimal-order}:
cubic for nonminimal order instead of quartic for minimal order.
See the summary in Figure~\ref{fig:complexity}, where the low
complexity of algorithms for minimal telescopers relies on Storjohann
and Villard's algorithms~\cite{Storjohann2005}, thus inducing a
\emph{certified\/} probabilistic feature.
We apply our results to the
calculation of diagonals in \S\,\ref{sec:implementation}, and describe our
implementation and comment on execution timings
in~\S\,\ref{sec:implementation}.

\begin{subsection}{Background on complexity --- Notation}
\label{sec:background} We recall basic notation and complexity facts
for later use. Let $k$ be again a field of characteristic zero.
Unless otherwise specified, all complexity estimates are given in
terms of arithmetical operations in $k$, which we denote by ``ops''.
Let $k[x]_{\leq d}^{m\times n}$ be the set of $m\times n$ matrices
with coefficients in~$k[x]$ of degree at most~$d$. Let
$\omega\in[2,3]$ be a feasible exponent of matrix multiplication, so
that two matrices from~$k^{n\times n}$ can be multiplied using
$\bigO(n^{\omega})$ ops.
Facts \ref{EvaInter} and~\ref{le:polymatrix} below
show the complexity of multipoint evaluation, rational
interpolation,
and algebraic operations on polynomial matrices using fast
arithmetic, where the notation $\bigOsoft(\cdot)$ indicates
cost estimates with hidden logarithmic
factors~\cite[Def.~25.8]{MCA2003}.

\begin{fact}\label{EvaInter}
For $p\in k[x]$ of degree less than~$n$, pairwise
distinct $u_0,\dots, u_{n-1}$ in~$k$, and $v_0, \dots, v_{n-1}\in k$,
we have:
\vspace{-0.15cm}
\begin{enumerate}
\item[(i)] Evaluating $p$ at the $u_i$'s takes $\bigOsoft(n)$ ops.
\vspace{-0.2cm}
\item[(ii)] For $m\in \{1, \dots, n\}$, constructing
$f=s/t\in k(x)$ with $\deg_x(s)<m$ and $\deg_x(t)\le n-m$ such that
$t(u_i)\neq 0$ and $f(u_i)=v_i$ for $0\le i\le n-1$ takes
$\bigOsoft(n)$~ops.
\end{enumerate}
\end{fact}

\begin{fact}\label{le:polymatrix}
For $M$ in $k[x]_{\leq d}^{m\times n}$, \ $d>0$, we have:
\vspace{-0.15cm}
\begin{enumerate}
  \item[(i)] If $M=\begin{pmatrix}M_1&M_2\end{pmatrix}$ is an invertible $n\times n$ matrix with $M_i\in k[x]_{\leq d_i}^{n\times n_i}$,
where $i=1,2$ and $n_1+n_2=n$, then the degree of $\det(M)$ is at most
$n_1d_1+n_2d_2$.
\vspace{-0.2cm}
  \item[(ii)] If $M=\begin{pmatrix}M_1&M_2\end{pmatrix}$ is not of full rank and with $M_i\in k[x]_{\leq d_i}^{m\times n_i}$, where $i=1,2$ and $n_1+n_2=n$, then there exists a nonzero $u\in
  k[x]^n$ with coefficients of degree at most~$n_1d_1+n_2d_2$ such that $Mu=0$.
\vspace{-0.2cm}
  \item[(iii)]The rank $r$ and a basis of the null space of $M$ can be computed using
  $\bigOsoft(nmr^{\omega-2}d)$~ops.
\end{enumerate}
\end{fact}

\noindent (For proofs, see \cite[Cor.~10.8, 5.18, 11.6]{MCA2003}
and~\cite[Th.~7.3]{Storjohann2005}.)

\smallskip

We call \emph{squarefree factorisation\/} of~$Q\in k[x, y]\setminus k[x]$
w.r.t.~$y$ the unique product $qQ_1Q_2^2\cdots Q_m^m$
equal to~$Q$ for $q\in k[x]$ and $Q_i\in k[x,y]$
satisfying~$\deg_y(Q_m)>0$ and such that the~$Q_i$'s are primitive,
squarefree, and pairwise coprime. The \emph{squarefree part~$Q^*$
of~$Q$\/} w.r.t.~$y$ is the product $Q_1 Q_2\cdots Q_m$.
Let $Q^-$ denote the polynomial~$Q/Q^*$, and $\lc_y(Q)$ the leading
coefficient of~$Q$ w.r.t.~$y$.
The following two formulas about $Q$, $Q^*$, and $Q^-$ can be proved
by mere calculations.

\begin{fact}\label{prop:deflation}
Let $\hat{Q}_i$ denote $Q^*/Q_i$. Then we have
\vspace{-0.15cm}
\begin{enumerate}
 \item[(i)]  ${Q^*D_y(Q^-)}/Q^- = \sum_{i=1}^m (i-1) \hat{Q}_i D_y(Q_i) \in k[x,
 y]$;
\vspace{-0.2cm}
 \item[(ii)] ${D_y(Q)}/Q^- = \sum_{i=1}^m i \hat{Q}_i D_y(Q_i)\in k[x, y]$.
\end{enumerate}
\end{fact}

Let $f=P/Q$ be a nonzero element in~$k(x, y)$, where~$P, Q$ are two
coprime polynomials in~$k[x, y]$. The degree of~$f$ in~$x$ is
defined to be~$\max\{\deg_x(P), \deg_x(Q)\}$, and denoted
by~$\deg_x(f)$. The degree of~$f$ in~$y$ is defined similarly. The
\emph{bidegree\/} of~$f$ is the pair~$(\deg_x(f),\deg_y(f))$, which
is denoted by~$\bideg(f)$. The bidegree~of~$f$ is said to be
\emph{bounded (above) by~$(\alpha,\beta)$}, written
$\bideg(f)\leq(\alpha,\beta)$, when~$\deg_x (f) \le\alpha$
and~$\deg_y (f) \le\beta$.

We say that~$f=P/Q$ is \emph{proper\/} if the degree of~$P$ in~$y$ is
less than that of~$Q$.
For creative telescoping, we may always assume w.l.o.g.\ that
$f=P/Q$~is proper. If not, rewrite $f = D_y(p) + \bar f$ with $p\in
k(x)[y]$ and $\bar f$~proper. A telescoper~$L$ for~$\bar f$ with
certificate~$\bar g$ is a telescoper for~$f$ with certificate~$L(p) + \bar g$.

\medskip\noindent{\bf Hypothesis (H)} \ \emph{From now on, $P$ and
$Q$ are assumed to be nonzero polynomials in $k[x, y]$ such that
$\deg_y(P)<\deg_y(Q)$, \ $\gcd(P, Q)=1$, and $Q$~is primitive
w.r.t.\ $y$.}

\medskip\noindent{\bf Notation} \ \emph{From now on, we write $(d_x,
d_y)$, $(d_x^*, d_y^*)$, and $(d_x^-, d_y^-)$ for the bidegrees of
$Q$, $Q^*$, and $Q^-$, respectively.}

\medskip\noindent The following hypothesis makes our estimates
concise.

\medskip\noindent{\bf Hypothesis (H')} \ \emph{Occasionally, we shall
require the ex\-tend\-ed hypothesis: Hypothesis~(H) and
$\deg_x(P)\le d_x$.}
\end{subsection}

\end{section}

\begin{section}{Hermite reduction}\label{sec:Hermite-reduction}
Let $K$ be a field of characteristic zero, either $k$ or~$k(x)$ in
what follows. Let $K(y)$ be the field of rational functions in~$y$
over~$K$, and $D_y$ be the usual derivation on it. For a rational
function~$f\in K(y)$, \emph{Hermite reduction\/}~\cite{Hermite1872}
computes rational functions $g$ and  $r=a/b$ in $K(y)$ satisfying
\begin{equation}\label{eq:ADP}
f = D_y(g) + r, \quad\deg_y(a)<\deg_y(b), \quad\text{$b$~is squarefree.}
\end{equation}
Horowitz and Ostrogradsky's method~\cite{Ostrogradsky1845,
Horowitz1971} computes the same decomposition as in~\eqref{eq:ADP}
by solving a linear system. For the details of those methods,
see~\cite[Chapter 2]{BronsteinBook}.
\begin{lemma}\label{le:unique}
If $f$~is proper, a pair $(g, r)$ satisfying~\eqref{eq:ADP} for
proper~$g,r$ is unique.
\end{lemma}
\begin{proof}
This is a consequence of~\cite[Theorem~2.10]{Horowitz1971} after
writing~$r$ as a sum~$\sum_{i=1}^m\alpha_i/(x-b_i)$ and integrating.
\end{proof}
\begin{lemma}\label{le:ComplexityUHR}
Let~$f$ be a nonzero rational function in~$K(y)$ of degree at most~$n$ in~$y$,
then Hermite reduction on~$f$ can
be performed using~$\bigOsoft(n)$ operations in~$K$.
\end{lemma}
\begin{proof}
See~\cite[Theorem 22.7]{MCA2003}.
\end{proof}
In contrast, the method of Horowitz and Ostrogradsky takes
$\bigO(n^\omega)$ operations in~$K$~\cite[\S\,22.2]{MCA2003}. Thus,
Hermite's method is quasi-optimal and asymptotically faster than the
former.

From now on, we fix $K=k(x)$ and analyse the complexity
of
Hermite reduction over~$k(x)$ in terms of
operations in~$k$.
To this end, we use an evaluation-interpolation approach.

\begin{subsection}{Output size estimates}
We derive an upper bound on the bidegrees of $g$ and~$r$
satisfying~\eqref{eq:ADP} by studying the linear system
in~\cite{Horowitz1971}.

Analysing Hermite reduction (under~(H))
shows the existence of $A,a \in k(x)[y]$
with $\deg_y(A)<d_y^-$, $\deg_y(a)<d_y^*$~and
\begin{equation}\label{eq:HOansatz}
\frac{P}{Q} = D_y\left(\frac{A}{Q^-}\right) + \frac{a}{Q^*}.
\end{equation}
In order to bound the bidegrees of $A$ and~$a$, we
reformulate~\eqref{eq:HOansatz} into the equivalent form
\begin{equation}\label{eq:HO}
P = Q^* D_y(A)-\left(\frac{Q^*D_y({Q^-})}{Q^-}\right) A + Q^- a,
\end{equation}
where $Q^*D_y(Q^-)/Q^-$ is a polynomial in $k[x, y]$ of bidegree at most $(d_x^*, d_y^*-1)$ by Fact~\ref{prop:deflation}.
Viewing $A$ and~$a$ as polynomials in $k(x)[y]$ with undetermined
coefficients, we form the following linear system, equivalent to~\eqref{eq:HO},
\begin{equation}\label{eq:HOsystem}
\begin{pmatrix}\HO_1&\HO_2\end{pmatrix}
\begin{pmatrix}\hat{A}\\\hat{a}\end{pmatrix} = \hat{P},
\end{equation}
where $\HO_1\in k[x]^{d_y \times d_y^-}_{\le d_x^*}$, $\HO_2\in
k[x]^{d_y \times d_y^*}_{\le d_x^-}$, and $\hat A$, $\hat a$, and
$\hat P$ are the coefficient vectors of $A$, $a$, and $P$ with sizes
$d_y^-$, $d_y^*$, and $d_y$, respectively.
Under the constraint of properness of $A/Q^-$ and~$a/Q^*$, $(A,a)$ is
unique by Lemma~\ref{le:unique}.
Then~\eqref{eq:HOsystem} has a unique solution, which leads to
the following lemma.
\begin{lemma}\label{le:HOsystem}
The matrix $\begin{pmatrix}\HO_1&\HO_2\end{pmatrix}$ is invertible over $k(x)$.
\end{lemma}
As the matrix $\begin{pmatrix}\HO_1&\HO_2\end{pmatrix}$ is
uniquely defined by~$Q$, we call it the matrix \emph{associated\/}
with~$Q$, denoted by $\HO(Q)$. Let $\delta$ be its determinant, so
that $\deg_x(\delta)\leq \mu:=d_x^*d_y^-+d_x^-d_y^*$ by
Fact~\ref{le:polymatrix}\emph{(i)}.
For later use, we also define $\delta'$~as the determinant
of~$\HO({Q^*}^2)$, so that $\deg_x(\delta')\leq \mu':=2d_x^*d_y^*$ by
Fact~\ref{le:polymatrix}\emph{(i)} and since $({Q^*}^2)^-=Q^*$.
\begin{lemma} \label{le:HRsize}
There exist $B, b\in k[x, y]$ with
$\deg_y(B) < d_y^-$ and~$\deg_y(b) < d_y^*$, and such that:
\vspace{-0.15cm}
\begin{enumerate}
\item[(i)] $\frac{P}{Q} = D_y \left(\frac{B}{\delta Q^-}\right) + \frac{b}{\delta Q^*}$;
\vspace{-0.2cm}
\item[(ii)] $\deg_x(B) \le \mu - d_x^* + \deg_x(P)$
and~$\deg_x(b)\le \mu - d_x^- + \deg_x(P)$.
\end{enumerate}
\end{lemma}
\begin{proof}
Applying Cramer's rule to~\eqref{eq:HOsystem}
leads to~\emph{(i)}.
Assertion~\emph{(ii)\/} next follows by determinant expansions.
\end{proof}

In what follows, we shall encounter proper rational functions
with denominator~$Q$ satisfying~$Q={Q^*}^2$.
The following lemma is an easy corollary of Lemma~\ref{le:HRsize} for
such functions.
\begin{cor} \label{cor:specialQ}
Assuming $Q={Q^*}^2$ in addition to Hypothesis~(H), there exist $B,
b\in k[x, y]$ with $\deg_y(B)$ and~$\deg_y(b)$ less than~$d_y^*$, and
such that
\begin{enumerate}
\item[(i)] $\frac{P}{{Q^*}^2} = D_y\left(\frac{B}{\delta'Q^*}\right) + \frac{b}{\delta'Q^*}$;
\item[(ii)] $\deg_x(B)$ and~$\deg_x(b)$ are bounded
by $\mu'-d_x^*+\deg_x(P)$.
\end{enumerate}
\end{cor}
\end{subsection}
\begin{subsection}{Algorithm by evaluation and interpolation}
We observe that an asymptotically optimal complexity can be achieved
by evaluation and interpolation at each step of Hermite reduction
over $k(x)$. This inspires us to adapt Gerhard's modular
method~\cite{Gerhard2001, Gerhard2004} to~$k(x, y)$.
Recall that, by Hyp.~(H), $Q\in k[x, y]$~is nonzero and primitive over~$k[x]$.

\medskip\noindent {\bf Definition}
\emph{An element $x_0\in k$ is \emph{lucky\/} if\/
$\lc_y(Q)(x_0)\neq 0$ and $\deg_y(\gcd(Q(x_0, y), D_y(Q(x_0,
y))))=d_y^-$.}
\begin{lemma}\label{le:unlucky}
There are at most $d_x(2d_y^*-1)$ unlucky points.
\end{lemma}
\begin{proof}
Let $\sigma\in k[x]$ be the $d_y^-$th subresultant w.r.t.\ $y$ of $Q$
and~$D_y(Q)$.
By~\cite[Corollary 5.5]{Gerhard2004}, all
unlucky points are in the set $U=\{\,x_0\in k \mid \sigma(x_0)=0\,\}$.
By~\cite[Corollary 3.2\emph{(ii)}]{Gerhard2004}, $\deg_x(\sigma)\le
d_x(2d_y^*-1)$.
\end{proof}
\begin{lemma}\label{le:commutative}
Let $B$, $b$, and $\delta$ be the same
as in Lemma~\ref{le:HRsize}, and let $x_0\in k$ be lucky.
Then $\delta(x_0)\neq 0$ and $(B(x_0, y), b(x_0, y))$
is the unique pair such that
\begin{equation}\label{eq:ADk}
\frac{P(x_0, y)}{Q(x_0, y)} =
D_y\left(\frac{B(x_0, y)}{\delta(x_0)Q^-(x_0, y)}\right) +
\frac{b(x_0, y)}{\delta(x_0) Q^*(x_0, y)}.
\end{equation}
\end{lemma}
\begin{proof}
By the luckiness of~$x_0$, $\deg_y (Q(x_0, y)) = d_y$ and
$Q(x_0,y)^-=Q^-(x_0,y)$, so $Q(x_0, y)^*=Q^*(x_0,y)$.
This implies $\HO(Q)(x_0,y)=\HO(Q(x_0,y))$, which, by
Lemma~\ref{le:HOsystem}, is invertible over~$k(x)$.
Hence~$\delta(x_0) \neq 0$,
and the evaluation at~$x=x_0$ of the equality in
Lemma~\ref{le:HRsize}\emph{(i)\/} is well-defined.
Thus, $(B(x_0, y), b(x_0, y))$ is a solution of~\eqref{eq:ADk}.
Uniqueness follows from Lemma~\ref{le:unique}.
\end{proof}
\begin{theorem}\label{th:BHR}
Algorithm \textsf{HermiteEvalInterp} in Figure~\ref{fig:HREvaInter} is
correct and takes~$\bigOsoft(d_xd_y^2+\deg_x(P)d_y)$~ops.
\end{theorem}
\begin{proof}
Set $\nu$ to~$d_x(2d_y^*-1)$. Lemma~\ref{le:unlucky} implies that
the $\lambda+1$ lucky points found in Step~3 are all less
than~$\lambda+\nu+1$.
By Lemmas \ref{le:unique} and~\ref{le:HRsize}\emph{(i)}, $A=B/\delta$
and~$a=b/\delta$.
By Lemma~\ref{le:commutative}, $A_0=B(x_0,y)/\delta(x_0)$ and
$a_0=b(x_0,y)/\delta(x_0)$.
By Lemma~\ref{le:HRsize}\emph{(ii)\/} and
since~$\deg_x(\delta)\leq\mu$, it suffices to rationally interpolate
$A$ and~$a$ from values at $\lambda+1$ lucky points.
This shows the correctness.
The dominant computation in Step~1 is the gcd, which takes
$\bigOsoft(d_xd_y)$~ops by \cite[Cor.~11.9]{MCA2003}.
For each integer $i\le \lambda +
\nu$, testing luckiness amounts to evaluations at~$x_0$ and
computing $\gcd(Q(x_0, y), D_y(Q(x_0, y)))$, which takes
$\bigOsoft(d_y)$ ops by Fact~\ref{EvaInter}\emph{(i)\/}
and~\cite[Cor.~11.6]{MCA2003}. Then,
generating~$S$ in Step~3 costs $\bigOsoft((\lambda + \nu+1)d_y)$ ops.
By Fact~\ref{EvaInter}\emph{(i)}, evaluations in Step~4 take
$\bigOsoft((\lambda + 1)d_y)$ ops. For each $x_0\in S$, the cost of
the Hermite reduction in Step~4 is $\bigOsoft(d_y)$ ops by
Lemma~\ref{le:ComplexityUHR}. Thus, the total cost of Step~4 is
$\bigOsoft((\lambda + 1)d_y)$ ops. By
Fact~\ref{EvaInter}\emph{(ii)}, Step~5
takes $\bigOsoft((\lambda + 1)d_y)$ ops.
Since $\lambda \le 2d_xd_y+
\deg_x(P)$ and $\nu \le 2d_xd_y$,
the total cost is as announced.
\end{proof}
%
\begin{figure}
\framebox[8.4cm]{
\begin{minipage}{8.1cm}
\rule[.3cm]{0cm}{0cm}
{\rm {Algorithm \textsf{HermiteEvalInterp}($P, Q$)}

\smallskip

\noindent \quad{\sc Input}:~$P, Q\in k[x, y]$ satisfying Hypothesis~(H).

\noindent \quad{\sc Output}:~$(A, a)\in k(x)[y]^2$
solving~\eqref{eq:HOansatz}.
\begin{enumerate}
\item Compute $Q^- := \gcd(Q, D_y(Q))$ and $Q^* := Q/Q^-$;
\item Set $\lambda := 2(d_x^*{d_y^-}+d_y^*d_x^-) + \deg_x(P)-\min\{d_x^-, d_x^*\}$;
\item Set $S$ to the set of $\lambda + 1$ smallest nonnegative integers that are lucky for~$Q$;
\item For each $x_0\in S$, compute $(A_0, a_0)\in k[y]^2$ such that
\[
\frac{P(x_0, y)}{Q(x_0, y)} = D_y\left(\frac{A_0}{Q^-(x_0,
y)}\right) + \frac{a_0}{Q^*(x_0, y)}
\]
using Hermite reduction over $k$;
\item Compute $(A, a)\in k(x)[y]$ by rational interpolation and return this pair.
\end{enumerate}}
\end{minipage}}
\caption{Hermite reduction over $k(x)$ via evaluation and
interpolation.} \label{fig:HREvaInter}
\vskip-10pt
\end{figure}

\end{subsection}

\vspace{-0.14cm}
As the generic output size of Hermite reduction is proportional
to~$\lambda d_y$, which is~$\bigO((d_xd_y+\deg_x(P))d_y)$, Algorithm
\textsf{HermiteEvalInterp} has quasi-optimal complexity.
\end{section}

\begin{section}{Minimal telescopers}\label{sec:minimal-order}
We analyse two algorithms for constructing minimal telescopers for
bivariate rational functions and their certificates.
\begin{subsection}{Hermite reduction approach}\label{HRA}
We design a new algorithm, presented in Figure~\ref{fig:HRTelescoping},
to compute minimal telescopers for rational functions by basing on
Hermite reduction. For~$f =P/Q \in k(x, y)$ and~$i \in \bN$, Hermite reduction
decomposes~$D_x^i(f)$ into
\begin{equation}\label{eq:ithHR}
D_x^i (f) = D_y (g_i) + r_i,
\end{equation}
where~$g_i, r_i \in k(x,y)$ are proper. Since the squarefree
part of the denominator of~$D_x^i(f)$ divides~$Q^*$, so does the
denominator of~$r_i$.
The following lemma shows that \eqref{eq:ithHR}~recombines into
telescopers and certificates;
next, Lemma~\ref{le:minimaltele} implies that the first pair obtained in
this way by Algorithm \textsf{HermiteTelescoping}
in Figure~\ref{fig:HRTelescoping}
yields a minimal telescoper.

\begin{lemma}\label{le:bound}
The rational functions $r_0, \dots, r_{d_y^*}$ are linearly dependent
over~$k(x)$.
\end{lemma}

\begin{proof}
The constraints on~$r_i$ imply
$\deg_y(r_iQ^*)< d_y^*$ for all $i\in \bN$, from which
follows the existence of a nontrivial linear dependence among the
$r_i$'s over $k(x)$.
\end{proof}

\begin{lemma}\label{le:minimaltele}
An integer~$\rho$ is minimal such that $\sum_{i=0}^\rho \eta_i r_i=0$
for $\eta_0, \ldots, \eta_\rho \in k(x)$ not all zero
if and only if
$\sum_{i=0}^\rho \eta_i D_x^i$ is a minimal telescoper for~$f$ with
certificate $\sum_{i=0}^\rho \eta_i g_i$.
\end{lemma}

\begin{proof}
Multiplying~\eqref{eq:ithHR} by $\eta_i$ before summing yields
\begin{equation*}
L(f) =
  D_y\biggl(\sum_{i=0}^\rho \eta_i g_i\biggr) + \sum_{i=0}^\rho \eta_ir_i
\quad\text{for}\quad
L := \sum_{i=0}^\rho \eta_i D_x^i ,
\end{equation*}
where the first two sums are proper.
Thus, by Lemma~\ref{le:unique}, $L$~is a telescoper of order~$\rho$ for~$f$
with certificate $\sum_{i=0}^\rho \eta_i g_i$
if and only if $\sum_{i=0}^\rho \eta_i r_i=0$ with~$\eta_\rho\neq0$.
The lemma follows.
\end{proof}

\begin{subsubsection}{Order bounds for minimal telescopers}\label{se:lowerbound}
Lemmas \ref{le:bound} and~\ref{le:minimaltele} combine into an upper bound
on the order of minimal telescopers for~$f$.
\begin{cor}\label{cor:upperbound}
Minimal telescopers have order at most~$d_y^*$.
\end{cor}
The bound~$6d_y$ is shown in~\cite{BCLSS2007}
for rational functions of the form $yD_y(Q)/Q$ with $Q\in k[x, y]$.
Apagodu and Zeilberger~\cite{Apagodu2006} obtain a similar bound for a
class of nonrational hyperexponential functions, but their proof does
not seem to apply to rational functions, as it heavily relies on the
presence of a nontrivial exponential part.

We also derive a lower bound on the order of the minimal telescoper,
to be used as an optimisation at the end of \S\,\ref{sec:compl-estim}:
choosing a lucky~$x_0\in k$,
next applying Hermite reduction in~$k(y)$ to~$D_x^i(f)(x_0, y)$, yields
\begin{equation}\label{eq:ithUHR}
D_x^i (f)(x_0, y) = D_y (g_{0, i}) + r_{0, i},
\end{equation}
where $g_{0, i}, r_{0, i}\in k(y)$ are proper and the denominator
of~$r_{0, i}$ divides~$Q^*(x_0, y)$.
Let $\rho_0$ be the smallest integer such that $r_{0, 0}, \dots, r_{0,
\rho_0}$ are linearly dependent over~$k$.
\begin{lemma}\label{le:lowerbound}
A minimal telescoper has order at least $\rho_0$.
\end{lemma}
\begin{proof}
We first claim that $r_{0, i}=r_i(x_0, y)$, for $r_i$ as
in~\eqref{eq:ithHR}. Note that the squarefree part w.r.t.~$y$ of the
denominator of $D_x^i(f)$ divides $Q^*$ for all $i\in\bN$.
By~\cite[Cor.~5.5]{Gerhard2004}, $x_0$ is lucky for the
denominator of $D_x^i(f)$ for all $i\in \bN$.
Then, the claim on~$r_{0,i}$ follows from Lemma~\ref{le:commutative}
applied to~$D_x^i(f)$.
Let $\rho$ be the
minimal order of a telescoper, then $r_0, \dots, r_{\rho}$ are
linearly dependent over~$k(x)$ by Lemma~\ref{le:minimaltele}. Thus
$r_{0, 0}, \dots, r_{0, \rho}$ are linearly dependent over~$k$,
which implies $\rho_0\le \rho$.
\end{proof}
\end{subsubsection}

\begin{subsubsection}{Degree bounds for minimal telescopers}
To derive degree bounds for $g_i$ and $r_i$
in~\eqref{eq:ithHR},
let $\delta$, $\delta'$, $\mu$, and $\mu'$ be defined
as before Lemma~\ref{le:HRsize},
and set $\mu'' = \mu + \mu' - 1$.

\begin{lemma} \label{le:special}
Let~$W$ be in $k[x, y]$ with $\deg_y(W) < d_y^*$. Then, for all $i
\in \bN$, there exist~$B, b\in k[x, y]$ with both $\bideg(B)$ and
$\bideg(b)$ bounded by~$(\deg_x(W) + \mu'', d_y^*-1)$, such that
\[ D_x\left(\frac{W}{{\delta}^{i+1} {\delta'}^i Q^*}\right) = D_y \left( \frac{B}{\delta ^{i+2} {\delta'}^{i+1} Q^*} \right)
+ \frac{b}{\delta ^{i+2} {\delta'}^{i+1} Q^*}.\]
\end{lemma}

\begin{proof}
A straightforward calculation leads to
\[  D_x \left(\frac{W}{{\delta}^{i+1}  {\delta'}^i Q^*}\right) = \frac{\tilde W}{{\delta}^{i+2} {\delta'}^{i+1} Q^*}
- \frac{1}{{\delta}^{i+1} {\delta'}^i} \frac{W D_x(Q^*)}{{Q^*}^2},
\]
where~$\bideg(\tilde{W}) \le (\deg_x(W) + \mu'', d_y^* -1)$.
By Corollary~\ref{cor:specialQ},
there exist $\tilde{B}, \tilde {b}\in k[x,y]$ such that
\[ \frac{1}{{\delta}^{i+1} {\delta'}^i} \frac{W D_x(Q^*)}{{Q^*}^2}
   = \frac{1}{{\delta}^{i+2} {\delta'}^{i+1}} \left( D_y \left(\frac{\delta  \tilde B }{Q^*} \right) + \frac{
{\delta} \tilde b }{Q^*} \right), \]
with $\bideg(\tilde{B})$ and $\bideg(\tilde{b})$ bounded by
$({\deg_x(W) + \mu'-1},$ ${d_y^*-1})$.
Setting~$(B, b) = (-{\delta} \tilde B, \tilde{W}- \delta \tilde b)$
ends the proof.
\end{proof}

\begin{lemma} \label{le:ithsize}
For~$i\in\bN$, there exist~$B_i, b_i \in k[x,y]$ such that
\begin{equation} \label{EQ:ten}
D_x^i(f) = D_y \left(\frac{B_i}{\delta^{i+1}{\delta'}^i {Q^*}^i Q^-}\right) + \frac{b_i}{\delta^{i+1}{\delta'}^i Q^*}.
\end{equation}
Moreover, $\bideg(B_i)\le (\deg_x(P) + \mu + i\mu'' + (i-1)
d_x^*, id_y^*+d_y^- -1)$ and $\bideg(b_i)\le(\deg_x(P) + \mu
+ i\mu''- d_x^-, d_y^* - 1)$.
\end{lemma}
\begin{proof}
We proceed by induction on~$i$.
For~$i=0$, the claim follows from
Lemma~\ref{le:HRsize}.
Assume that~$i>0$ and that the claim holds
for the values less than~$i$.
For brevity, we set $\gamma= \deg_x(P) + \mu$,
\ $F_{i-1} = B_{i-1}/({\delta}^i{\delta'}^{i-1} {Q^*}^{i-1} Q^-)$,
and $G_{i-1} = b_{i-1}/({\delta}^i{\delta'}^{i-1} Q^*)$.
The induction hypothesis implies
\[D_x^i(f) = D_y D_x (F_{i-1}) + D_x (G_{i-1}),\]
with bidegree bounds on $B_{i-1}$ and~$b_{i-1}$.
Fact~\ref{prop:deflation}\emph{(i)\/} implies that $\tilde{Q} :=
Q^*D_x(Q^-)/Q^-$ is in $k[x, y]$, with
$\bideg(\tilde{Q}) \le (d_x^*-1, d_y^*)$.
Hence $D_x(1/Q^-)=-\tilde{Q}/Q$.
This observation and an easy calculation imply that
\[D_x (F_{i-1}) = \frac{\tilde{B}_{i-1}}{\delta^{i+1}{\delta'}^i {Q^*}^i Q^-},\]
where~$\tilde{B}_{i-1} \in k[x, y]$ and~$\deg_x(
\tilde{B}_{i-1}) \le \deg_x(B_{i-1}) + \mu'' + d_x^*$.
Furthermore, by Lemma~\ref{le:special}
there are~$\bar B_i, \bar b_i \in k[x, y]$ with bidegrees at most
$(\deg_x(b_{i-1}) + \mu'', d_y^*-1)$, such that
\[ D_x(G_{i-1}) = D_y \left( \frac{\bar B_i}{{\delta}^{i+1}{\delta'}^i
Q^*} \right) + \frac{\bar b_i}{{\delta}^{i+1}{\delta'}^i Q^*}. \]
Setting~$B_i=\tilde{B}_{i-1}+\bar B_i{Q^*}^{i-1}Q^-$ and~$b_i=\bar b_i$, we
arrive at~\eqref{EQ:ten}.
It remains to verify the degree bounds.
The induction hypothesis implies that
both $\deg_x(\bar B_i)$ and~$\deg_x(b_i)$ are bounded by~$\gamma + i \mu''  -
d_x^-$.
It follows that $\deg_x(\bar B_i{Q^*}^{i-1}Q^-)$ is
bounded by~$\gamma + i \mu'' + (i-1)d_x^*$.
Similarly,~$\deg_x(\tilde{B}_{i-1})$ is bounded by~$\gamma + i \mu'' +
(i-1)d_x^*$, and so is~$\deg_x (B_i)$.
The bounds on degrees in~$y$ are obvious.
\end{proof}

We next derive degree bounds for the minimal telescopers
obtained at an intermediate stage of \textsf{HermiteTelescoping};
refined bounds on the output will be given by Theorem~\ref{th:AZtelesize}.

\begin{lemma}
Under~(H'), Step~2(c) of Algorithm \textsf{Hermite\-Telescoping} computes
a minimal telescoper $L \in k[x]\langle D_x \rangle$ with order~$\rho$
and a certificate $g \in k(x, y)$ for $P/Q$ with $\deg_x(L)\in
\bigO(d_xd_y\rho^2) $ and $\bideg(g) \in \bigO(d_xd_y\rho^2) \times
\bigO(d_y\rho)$.
\end{lemma}

\begin{proof}
By Lemma~\ref{le:minimaltele}, we exhibit a minimal telescoper
by considering the first nontrivial linear dependence among the
$a_i$'s in~\eqref{EQ:ten}.
Let $M$ be the coefficient matrix of the system in~$(\eta_i)$
obtained from $\sum_{i=0}^{\rho}\eta_i a_i=0$.
By Lemma~\ref{le:ithsize}, $M$~is of size at most
$(\rho+1)\times d_y^*$ and with coefficients of degree at most
$\sigma := d_x+\mu+\rho \mu''-d_x^-$ in $x$.
Hence, there exists a solution
$(\eta_0, \dots, \eta_{\rho})\in k[x]^{\rho+1}$ of degree at
most~$\sigma\rho$ in~$x$ by
Fact~\ref{le:polymatrix}\emph{(ii)}.
Since $\mu, \mu''\in
\bigO(d_xd_y)$ and $d_y^*\le d_y$, the degree estimates of
$L$ and $g$ are as announced.
\end{proof}
\end{subsubsection}

\begin{figure}
\framebox[8.4cm]{
\begin{minipage}{8.1cm}
\rule[.3cm]{0cm}{0cm}
{\rm {Algorithm \textsf{HermiteTelescoping}($f$)}

\smallskip

\noindent \quad{\sc Input}: $f=P/Q\in k(x, y)$
satisfying Hypothesis~(H).

\noindent \quad{\sc Output}:
A minimal telescoper $L\in k[x]\langle D_x \rangle$ with certificate $g\in k(x, y)$.
\begin{enumerate}
\item Apply~\textsf{HermiteEvalInterp} to $f$ to get $(g_0, a_0)$ such that $f = D_y(g_0) + a_0/Q^*$.
If $a_0=0$, return $(1, g_0)$.
\item For $i$ from 1 to $\deg_y(Q^*)$ do
\begin{enumerate}
\item Apply~\textsf{HermiteEvalInterp} to $-a_{i-1}D_x(Q^*)/{Q^*}^2$
  to express it as $D_y(\tilde{g}_{i})+ \tilde{a}_{i}/Q^*$.
\item Set $g_i = D_x(g_{i-1}) + \tilde{g}_i$ and $a_i = D_x(a_{i-1}) + \tilde{a}_i$.
\item Solve $\sum_{j=0}^i \eta_j a_j =0$ for $\eta_j\in k(x)$
  using \cite{Storjohann2005}.
If there exists a
nontrivial solution, then set $(L,g) := \bigl(\sum_{j=0}^i \eta_j D_x^j,
\sum_{j=0}^i \eta_jg_j\bigr)$, and break.
\end{enumerate}
\item Compute the content~$c$ of~$L$ and return
$(c^{-1}L, c^{-1}g)$.
\end{enumerate}}
\end{minipage}}
\caption{Creative telescoping by Hermite reduction}
\label{fig:HRTelescoping}
\vskip-10pt
\end{figure}

\begin{subsubsection}{Complexity estimates}\label{sec:compl-estim}
We proceed to analyse the complexity of the algorithm in
Figure~\ref{fig:HRTelescoping} and of an optimisation.

\begin{theorem}
Under Hyp.~(H'), Algorithm
\textsf{HermiteTelescoping} in Figure~\ref{fig:HRTelescoping} is correct
and takes $\bigOsoft(\rho^{\omega+1}d_x d_y^2)$ ops,
where $\rho$ is the order of the minimal telescoper.
\end{theorem}

\begin{proof}
The formulas in Step~2(a) create the loop invariant $D_x^i(f) =
D_y(g_i) + a_i/Q^*$.
Correctness then follows from Lemmas~\ref{le:bound} and~\ref{le:costlowerbound}.
Step~1 takes $\bigOsoft(d_xd_y^2)$ ops by Theorem~\ref{th:BHR} under~(H').
By Lemma~\ref{le:ithsize}, $\deg_x(-a_{i-1}D_x(Q^*))\in\bigO(id_xd_y)$.
So the cost for performing Hermite reduction on
$-a_{i-1}D_x(Q^*)/{Q^*}^2$ in Step~2(a) is~$\bigOsoft(id_xd_y^2)$
ops by Theorem~\ref{th:BHR}. The bidegrees of $g_i$ and $a_i$ in
Step~2(b) are in $\bigO(id_xd_y) \times \bigO(id_y)$ by
Lemma~\ref{le:ithsize}. Since adding and differentiating have linear
complexity, Step~2(b) takes $\bigOsoft(i^2d_xd_y^2)$ ops.
For each~$i$, the coefficient matrix of $\sum_{j=0}^i \eta_j a_j =0$
in Step~2(c) is of size at most $(i+1)\times d_y^*$ and with
coefficients of degree at most $\deg_x(a_i)\in \bigO(id_xd_y)$.
Moreover, the rank of this matrix is either $i$ or~$i+1$.
Then, Step~2(c) takes $\bigOsoft(i^{\omega}d_x d_y^2)$ ops by
Fact~\ref{le:polymatrix}\emph{(iii)}.
Computing the content and divisions in Step~3 has
complexity $\bigOsoft(d_xd_y\rho^3)$.
If the algorithm returns when
$i=\rho$, then the total cost is in
\begin{equation}\label{eq:analyis-without-rho0}
\sum_{i=0}^{\rho} \bigOsoft(i^2d_xd_y^2) +
\sum_{i=1}^{\rho} \bigOsoft(i^{\omega} d_x d_y^2)
\subset \bigOsoft(\rho^{\omega+1}d_x d_y^2)~\text{ops},
\end{equation}
which is as announced.
\end{proof}

An optimisation, based on Lemma~\ref{le:lowerbound},
consists in guessing the order~$\rho$ so as to perform
Step~2(c) a few times only:
As a preprocessing step, choose $x_0\in k$ lucky for~$Q$, then detect
linear dependence of~$\{r_{0,0},\dots,r_{0,j}\}$ in~\eqref{eq:ithUHR}.
The minimal~$j$ for dependence is a lower bound~$\rho_0$ on~$\rho$.
So Step~2(c) is then performed only when~$i\ge\rho_0$.
In practice, the lower bound~$\rho_0$ computed in this way almost always
coincides with the actual order~$\rho$.
So normalising the $g_i$'s becomes the dominant step, as observed
in experiments.
We analyse this optimisation by first estimating the cost for
computing~$\rho_0$.

\begin{lemma}\label{le:costlowerbound}
Under Hypothesis~(H'), computing a lower order bound~$\rho_0$ for minimal
telescopers takes $\bigOsoft(d_xd_y\rho_0^3)$~ops.
\end{lemma}

\begin{proof}
Since differentiating has linear complexity, the derivative
$D_x^i(f)$ takes $\bigOsoft(i^2d_xd_y)$ ops.
By Fact~\ref{EvaInter}\emph{(i)}, the evaluation $D_x^i(f)(x_0, y)$
takes as much.
The cost of Hermite reduction
on $D_x^i(f)(x_0, y)$ is $\bigOsoft(id_y)$ ops by
Lemma~\ref{le:ComplexityUHR}.
By Fact~\ref{le:polymatrix}\emph{(iii)\/} with~$d=1$,
computing the rank of the coefficient
matrix of~$\sum_{j=0}^i\eta_j r_{0, j}$, with $r_{0, j}$ as
in~\eqref{eq:ithUHR}, takes $\bigOsoft(d_y i^{\omega-1})$ ops.
Thus, the total cost for computing a lower bound on~$\rho_0$
is $\sum_{i=0}^{\rho_0} \bigOsoft(i^2 d_xd_y)\in
\bigOsoft(d_xd_y\rho_0^3)$~ops.
\end{proof}

\begin{cor}
For runs such that $\rho_0=\rho-\bigO(1)$, the previous optimisation
of~\textsf{HermiteTelescoping} takes
$\bigOsoft(\rho^3d_xd_y^2)$~ops.
\end{cor}
\begin{proof}
In view of Lemma~\ref{le:costlowerbound},
the estimate \eqref{eq:analyis-without-rho0} becomes
$\bigOsoft(d_xd_y\rho_0^3) +
\sum_{i=0}^{\rho} \bigOsoft(i^2d_xd_y^2) +
\sum_{i=\rho_0}^{\rho} \bigOsoft(i^{\omega} d_x d_y^2)$,
which is $\bigOsoft(\rho^3 d_x d_y^2) +
\bigOsoft((\rho-\rho_0)\rho^\omega d_x d_y^2)$~ops,
whence the result.
\end{proof}

\end{subsubsection}
\end{subsection}

\begin{subsection}{Almkvist and Zeilberger's approach}\label{AZA}
We analyse the complexity of Almkvist and Zeilberger's
algorithm~\cite{Almkvist1990} when restricted to bivariate rational
functions.
In order to get a telescoper whose order~$\rho$ is minimal, the resulting
algorithm, denoted \textsf{RatAZ}, solves~\eqref{eq:CT} for
increasing, prescribed values of~$\rho$
until it gets a solution $(\eta_0, \dots, \eta_{\rho}, g)\in
k(x)^{\rho+1}\times k(x, y)$ with the $\eta_i$'s not all zero.
For the analysis, we start by studying the parameterisation of the
differential Gosper algorithm of~\cite{Almkvist1990}
under the same restriction to~$k(x,y)$.

\begin{definition}[\cite{Gerhard2004}]
Let $K$ be a field and $a, b\in K[y]$ be non\-zero polynomials. A
triple $(p, q, r)\in K[y]^3$ is said to be a \emph{differential Gosper
form\/} of the rational function $a/b$ if
\[\frac{a}{b} = \frac{D_y(p)}{p} + \frac{q}{r}~\text{and $\gcd(r,
  q-\tau D_y(r))=1$ for all $\tau\in \bN$}.\]
\end{definition}

For hyperexponential~$f$, a key step in~\cite{Almkvist1990} is to
\emph{compute\/} a differential Gosper form of the logarithmic derivative of
$F=\sum_{i=0}^{\rho} \eta_iD_x^i(f)$, where the~$\eta_i$'s are
undetermined from~$k(x)$.
In the analogue \textsf{RatAZ}, this form is \emph{predicted\/} by
Lemma~\ref{le:dGf-for-F'/F} below, which is a technical
generalisation of a result by Le~\cite{Le2000} on~$F$ when~$f$ has a
squarefree denominator.

Write $Q=t(y)T(x,y)$, splitting content and primitive part
w.r.t.~$x$.
By an easy induction, $D_x^i(f)=N_i/(Q{T^*}^i)$ for $N_i\in k[x, y]$.
For this section, set $F=\sum_{i=0}^{\rho} \eta_iD_x^i(f)$, $N=\sum_{i=0}^{\rho}\eta_iN_i{T^*}^{\rho-i}$, and
$H=-D_y(Q)/Q^--\rho t^*D_y(T^*)$.

\begin{lemma}\label{le:dGf-for-F'/F}
If $F$~is nonzero,
the triple $(N, H, Q^*)$ is a differential Gosper form of $D_y(F)/F$.
\end{lemma}

\begin{proof}
First, observe $F=N/(Q{T^*}^\rho)$ and~$Q^*=t^*T^*$.
Next, $D_y(F)/F = D_y(N)/N - D_y(Q)/Q - \rho D_y(T^*)/T^*$ is
$D_y(N)/N + H/Q^*$.
There remains to prove $\gcd(Q^*, H - \tau D_y(Q^*)) = 1$, for
any~$\tau \in \bN$.
Recall that the squarefree part~$Q^*$ of~$Q$ is the product~$Q_1Q_2\cdots Q_m$
and that $\hat{Q}_i$~denotes~$Q^*/Q_i$.
By Fact~\ref{prop:deflation}\emph{(ii)},
\begin{equation*}
Z := H - \tau D_y(Q^*)
  = -\rho t^*D_y(T^*)-\sum_{i=1}^m(i+\tau)\hat Q_iD_y(Q_i) .
\end{equation*}
If $Q_j$~divides~$t^*$, $Z$~reduces to~$-(j+\tau)\hat Q_jD_y(Q_j)$ modulo~$Q_j$.
If not, it reduces to $-(j+\tau)\hat Q_jD_y(Q_j) -\rho
t^*(D_y(Q_j)T^*/Q_j)$,
which rewrites to~$-(j+\tau+\rho)\hat Q_jD_y(Q_j)$ modulo~$Q_j$.
In both cases, $Z$~is coprime with~$Q^*$, as $j>0$, \ $\tau\geq0$,
and~$\rho\geq0$.
\end{proof}

By another induction, we observe
$\bideg(N_i) \le (\deg_x(P)+i \deg_x(T^*)-i, d_y+i \deg_y(T^*)-1)$,
so that
$\bideg(N) \le (\deg_x(P)+\rho \deg_x(T^*)-\rho, d_y+\rho \deg_y(T^*)-1)$.

The next step in \textsf{RatAZ} is, for fixed~$\rho$, to reduce~\eqref{eq:CT}
by the change of unknown $g=z/(Q^-{T^*}^\rho)$, so as to determine all
$(\eta_i)\in k(x)^{\rho+1}$ for which the differential equation in~$z$
\begin{equation}\label{eq:pdga}
\sum_{i=0}^{\rho}\eta_iN_i{T^*}^{\rho-i}=Q^*D_y(z)+\left(D_y(Q^*)+H\right)z
\end{equation}
has a polynomial solution in $k(x)[y]$.
For later use, we recall the following consequence
of~\cite[Corollary 9.6]{Gerhard2004}.

\begin{lemma}\label{le:polysol}
Let $a, b\in K[y]$ be such that $\beta=-\lc_y(b)/\lc_y(a)$ is a
nonnegative integer and $\deg_y(b)=\deg_y(a)-1$.
Let $c\in K[y]$ be such that $\beta\ge \deg_y(c)-\deg_y(a)+1$.
If $u$~is a polynomial solution of $aD_y(z)+bz=c$, then $\deg_y(u)\le
\beta$.
\end{lemma}
The following lemma generalises~\cite[Lemma~2]{Le2000}
to present a degree bound for~$z$.


\begin{figure}
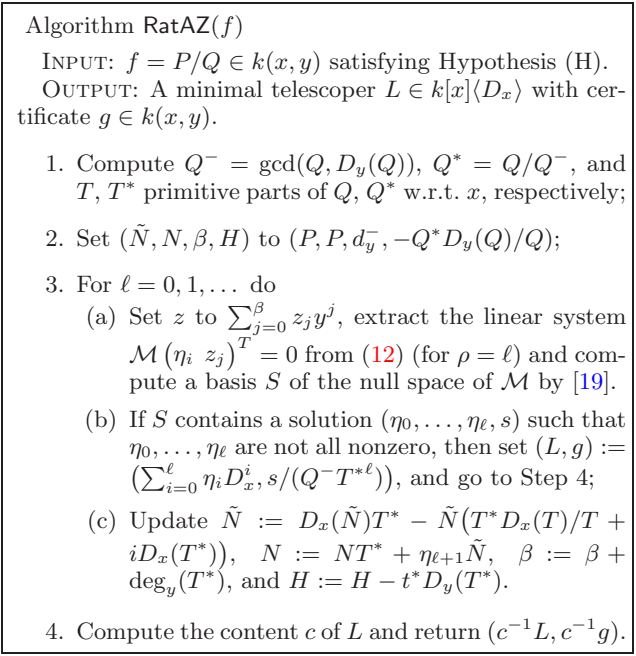

\framebox[8.4cm]{
\begin{minipage}{8.1cm}
\rule[.3cm]{0cm}{0cm}
{\rm {Algorithm \textsf{RatAZ}($f$)}

\smallskip

\noindent \quad{\sc Input}:
$f=P/Q\in k(x, y)$ satisfying Hypothesis~(H).

\noindent \quad{\sc Output}:
A minimal telescoper $L\in k[x]\langle D_x \rangle$ with certificate $g\in k(x, y)$.
\begin{enumerate}
\item Compute $Q^-=\gcd(Q, D_y(Q))$, $Q^*=Q/Q^-$,
and $T$, $T^*$ primitive parts of $Q$, $Q^*$ w.r.t.~$x$,
respectively;
\item Set $(\tilde N, N, \beta, H)$ to $(P, P, d_y^-, -Q^*D_y(Q)/Q)$;
\item For $\ell=0, 1, \dots$ do
\vspace{-0.2cm}
\begin{enumerate}
\item Set $z$ to~$\sum_{j=0}^{\beta} z_jy^j$, extract the linear system
$\cM\begin{pmatrix}\eta_i&z_j\end{pmatrix}^T=0$
from~\eqref{eq:pdga} (for~$\rho=\ell$) and
compute a basis~$S$ of the null space of~$\cM$ by~\cite{Storjohann2005}.
\item If $S$ contains a solution $(\eta_0, \dots,
\eta_{\ell}, s)$ such that $\eta_0, \dots, \eta_{\ell}$ are not all
nonzero, then set $(L,g) := \bigl(\sum_{i=0}^{\ell} \eta_iD_x^i,
s/(Q^-{T^*}^\ell)\bigr)$, and go to Step~4;
\item Update $\tilde N := D_x(\tilde N)T^*
- \tilde{N}\bigl(T^*D_x(T)/T+iD_x(T^*)\bigr)$, \ $N := NT^* +
\eta_{\ell+1}\tilde{N}$, \ $\beta := \beta + \deg_y(T^*)$, and $H := H
- t^*D_y(T^*)$.
\end{enumerate}
\item Compute the content~$c$ of~$L$ and return
$(c^{-1}L, c^{-1}g)$.
\end{enumerate}
}
\end{minipage}}
\caption{Improved Almkvist--Zeilberger algorithm}
\label{fig:RatAZ}
\vskip-10pt
\end{figure}

%
\begin{lemma}\label{le:degb}
If $u\in k(x)[y]$ is a solution of~\eqref{eq:pdga} for $(\eta_i)\in
k(x)^{\rho+1}$, then\/ $\deg_y(u)$ is bounded by $\beta = d_y^- + \rho
\deg_y(T^*)$.
\end{lemma}
\begin{proof}
Let $a=Q^*$ and $b=D_y(Q^*)+ H$.
By the definition of~$H$,
$b =-Q^*D_y(Q^-)/Q^- - \rho t^*D_y(T^*)$.
Fact~\ref{prop:deflation}\emph{(i)\/} implies that
$\lc_y(b)=-(d_y^- + \rho \deg_y(T^*))\lc_y(a)$.
Therefore, $\beta =
-\lc_y(b)/\lc_y(a)= d_y^- + \rho \deg_y(T^*)$.
As $\deg_y(N)< d_y
+ \rho\deg_y(T^*)$ and $d_y = d_y^* + d_y^-$, \ $\beta\ge
\deg_y(N)-d_y^*+1$.
The lemma holds by Lemma~\ref{le:polysol}.
\end{proof}

We end the present section using the approach of Almkvist and Zeilberger
to provide tight degree bounds on the outputs from Algorithms
\textsf{HermiteTelescoping} and \textsf{RatAZ}.

\begin{theorem}\label{th:AZtelesize}
Under Hypothesis~(H'), there exists a minimal telescoper $L \in
k[x]\langle D_x \rangle$ with certificate $g \in k(x, y)$ with
$\deg_x(L)\in \bigO(d_xd_yd_y^*) $ and $\bideg(g) \in \bigO(d_xd_yd_y^*)
\times \bigO(d_yd_y^*)$.
\end{theorem}

\begin{proof}
By Corollary~\ref{cor:upperbound}, there exists a
smallest $\rho\in \bN$ at most $d_y^*$, for which \eqref{eq:CT}~has a
solution with the $\eta_i$'s not all zero.
For this $\rho$, we estimate the size of the polynomial
matrix~$\cM$ derived from~\eqref{eq:pdga} by undetermined
coefficients.
By the remark on~$N$ after
Lemma~\ref{le:dGf-for-F'/F}, we have $\bideg(N) \le (n_x,n_y)$ where
$n_x := d_x + \rho \deg_x(T^*)-\rho\in \bigO(\rho d_x)$ and
$n_y := d_y + \rho \deg_y(T^*)-1\in \bigO(\rho d_y)$.
The matrix~$\cM$ contains two blocks
$\cM_1\in k[x]_{\leq n_x}^{(n_y+1)\times(\rho+1)}$ and $\cM_2\in k[x]_{\leq
d_x}^{(n_y+1)\times(\beta+1)}$, where $\beta\in \bigO(\rho d_y)$ is the
same as in Lemma~\ref{le:degb}.
By the minimality of $\rho$, the dimension of the null space of~$\cM$
is~1. So there exists $u\in k[x]^{n_y+1}$ with coefficients of
degree at most $n_x(\rho+1)+d_x(\beta+1)\in \bigO(d_x d_yd_y^*)$ in~$x$
such that $\cM\begin{pmatrix}\eta&z\end{pmatrix}^T=0$, which implies degree bounds in $x$ for $L$ and~$g$.
The degree bound in $y$ for $g$ is
obvious.
\end{proof}

We now analyse the complexity of the algorithm in
Fig.~\ref{fig:RatAZ}.

\begin{theorem}
Under Hypothesis~(H'),
Algorithm~\textsf{RatAZ} in Figure~\ref{fig:RatAZ} is correct and takes
$\bigOsoft(d_xd_y^{\omega}\rho^{\omega+2})$ ops, where $\rho$ is the
order of the minimal telescoper.
\end{theorem}
\begin{proof}
By the existence of a telescoper, Corollary~\ref{cor:upperbound}, and
Lemma~\ref{le:degb}, the algorithm
always terminates and returns a minimal telescoper~$L$,
of order~$\rho$ at most~$d_y^*$.
Gcd computations dominate the cost of Steps 1 and~2, which
take $\bigOsoft(d_xd_y^2)$ ops.
For each $\ell \in \bN$, the dominating cost in Step~3 is
computing the null space of~$\cM$.
Let $n_y = d_y + \ell \deg_y(T^*)-1\in \bigO(\ell d_y)$ and $n_x=
d_x + \ell \deg_x(T^*)\in \bigO(\ell d_x)$. By the same argument as
in the proof of Theorem~\ref{th:AZtelesize}, the matrix $\cM$ is of
size at most $(n_y+1)\times (\ell + \beta + 2)$ and with coefficients of
degree at most~$n_x$.
Let $r$ be the rank of~$\cM$, which is either $\ell+\beta+2$ or
$\ell+\beta+1$ by construction.
Thus, a basis of the null space of~$\cM$ can
be computed within $\bigOsoft(n_x(n_y+1)(\ell+\beta+2)r^{\omega-2})$ ops
by Fact~\ref{le:polymatrix}\emph{(iii)}. Since $\beta\in \bigO(\ell
d_y)$, $\bigOsoft(n_x (n_y+1)(\ell+\beta+2)r^{\omega-2})$ is included in
$\bigOsoft(d_xd_y^{\omega}{\ell}^{\omega+1})$.
Since Step~3
terminates at $\ell=\rho$, the total cost of the algorithm is
$\sum_{\ell=0}^{\rho} d_xd_y^{\omega}{\ell}^{\omega+1}$ ops.
This is within the announced complexity,
$\bigOsoft(d_xd_y^{\omega}\rho^{\omega+2})$~ops.
\end{proof}

\begin{cor}
Algorithms \textsf{HermiteTelescoping} and \textsf{RatAZ} in
Fig.\ \ref{fig:HRTelescoping} and~\ref{fig:RatAZ} both output the
primitive minimal telescoper~$L$ together with its certificate~$g$,
which satisfy $\deg_{D_x}(L)\leq d_y^*$,
\ $\deg_x(L),\deg_x(g)\in\bigO(d_xd_yd_y^*)$, and
$\deg_y(g)\in\bigO(d_yd_y^*)$.
\end{cor}

\begin{proof}
Both algorithms output the primitive minimal telescoper, as they
compute a minimal telescoper at an intermediate step, and
owing to their last step of content removal.
Bounds follow from Corollary~\ref{cor:upperbound} and
Theorem~\ref{th:AZtelesize}.
\end{proof}

\end{subsection}

\end{section}

\begin{section}{Nonminimal telescopers}\label{sec:nonminimal-order}
Here, we discard Hypothesis~(H) and trade the minimality of telescopers for smaller total output sizes.
To this end, we adapt and slightly extend the arguments
in~\cite{Lipshitz1988} and~\cite[\S\,3]{BCLSS2007}.

\medskip
Given $f=P/Q\in k(x, y)$ of bidegree~$(d_x, d_y)$, our goal is to
find a (possibly nonminimal) telescoper for~$f$.
It is sufficient to
find a nonzero differential operator $A(x,D_x,D_y)$ that annihilates~$f$.
Indeed, any $A\in k[x]\langle D_x, D_y\rangle \setminus \{0\}$ such that
$A(f)=0$ can be written $A= D_y^r(L + D_y R)$, where $L$ is
nonzero in $k[x]\langle D_x \rangle$ and $R\in k[x]\langle D_x, D_y\rangle$.
If $r=0$, then clearly $L$ is a telescoper for $f$; otherwise, $A(f)=0$ yields
$L(f) = D_y(-R(f)-\sum_{i=0}^{r-1}\frac{a_i}{i+1}y^{i+1})$
for some $a_i\in k(x)$,
which implies that $L$ is again a telescoper for~$f$. Moreover,
in both cases, $\deg_x(L)\le \deg_x(A)$ and $\deg_{D_x}(L)\le \deg_{D_x}(A)$.
Furthermore, for any $(i, j, \ell) \in \bN^3$, a direct calculation yields
\begin{equation}\label{eq:nonminimal}
x^iD_x^jD_y^{\ell}(f)=\frac{H_{i,j,\ell}}{Q^{j+\ell+1}},
\end{equation}
where $H_{i, j, \ell}\in k[x, y]$ and $\deg_x(H_{i,j,\ell})\leq
(j+\ell+1)d_x+i-j$ and $\deg_y(H_{i,j,\ell})\leq
(j+\ell+1)d_y-\ell$.  From these inequalities, we derive
the size and complexity estimates in Figure~\ref{fig:complexity} (bottom half), using
two different filtrations of $k[x]\langle D_x, D_y\rangle$.

\medskip\noindent{\bf Lipshitz's filtration (\cite{Lipshitz1988}).}
Let $F_\nu$ be the $k$-vector space
of dimension ${\sf f}_\nu:=\binom{\nu+3}{3}$
spanned by $\{\,x^iD_x^jD_y^{\ell} \mid i+j+\ell
\leq \nu\,\}$.
By~\eqref{eq:nonminimal}, $F_\nu(f)$ is contained
in the vector space of dimension
${\sf g}_\nu:=\left((\nu+1) d_x+\nu+1\right)\left((\nu+1) d_y+1\right)$
spanned by
$\bigl\{\,\frac{x^i y^j}{Q^{\nu+1}} \mid i \le (\nu+1) d_x+\nu, \ j \le (\nu+1) d_y\,\bigr\}$.
Choosing
$\nu=6(d_x+1)(d_y+1)$ yields ${\sf f}_\nu > {\sf g}_\nu$;
therefore, there exists $A$ in $k\langle x, D_x, D_y\rangle
\setminus \{0\}$ with total degree at most
$6(d_x+1)(d_y+1)$ in $x$,
$D_x$, and $D_y$ that annihilates~$f$.
Moreover, $A$~is found by
linear algebra in dimension
$\bigO((d_x d_y)^3)$.

\smallskip\noindent{\bf A better filtration (\cite{BCLSS2007}).}
Instead of taking total
degree, set $F_{\kappa, \nu}$ to the $k$-vector space of dimension
${\sf f}_{\kappa,\nu}:=(\kappa+1)\binom{\nu+2}{2}$
generated by
$\{\, x^iD_x^jD_y^{\ell} \mid i\le \kappa, \ j+\ell
\leq \nu\,\}$.
By~\eqref{eq:nonminimal},  $F_{\kappa, \nu}(f)$ is contained
in the vector space of dimension
\hbox{${\sf g}_{\kappa,\nu}:=((\nu+1) d_x+\kappa+1)((\nu+1) d_y+1)$ spanned by}

\noindent
$\bigl\{\,\frac{x^i y^j}{Q^{\nu+1}} \mid i \le (\nu+1) d_x+\kappa, \ j \le (\nu+1) d_y\,\bigr\}$.
Choosing $\kappa=3d_xd_y$ and $\nu=6d_y$ results in ${\sf f}_{\kappa,\nu} > {\sf g}_{\kappa,\nu}$.
This implies the existence of $A$ in $k\langle x, D_x, D_y\rangle
\setminus \{0\}$ with total degree at most $6d_y$ in $D_x$ and $D_y$
and degree at most $3d_x d_y$ in~$x$ that annihilates $f$.
Again, $A$ is found by linear algebra over~$k$, but in smaller dimension
$\bigO(d_x d_y^3)$.

\end{section}

\begin{section}{Implementation and timings}\label{sec:implementation}

We implemented in Maple~13 all the algorithms described;
as we used Maple's generic solver \verb+SolveTools:-Linear+,
all of our implementations are deterministic.

The evaluation-interpolation algorithm \textsf{HermiteEvalInterp}
for Hermite reduction (Fig.~\ref{fig:HREvaInter}) does not
perform well, mainly because Maple's rational interpolation routines
are far too slow.
We thus implemented Algorithm~\textsf{HermiteReduce} (original
version) in~\cite[\S\,2.2]{BronsteinBook} (carefully avoiding
redundant extended gcd calculations), and noted that it performs
better.

We then implemented a variant of Algorithm \textsf{HermiteTelescoping}
in Figure~\ref{fig:HRTelescoping}, using \textsf{HermiteReduce} in
place of \textsf{HermiteEvalInterp}, and including the optimisation at
the end of \S\,\ref{sec:compl-estim},
refined by additional modular calculations.

For a rational function,
Algorithm \textsf{HermiteTelescoping} returns the minimal telescoper~$L$ and the certificate~$g$.
The algorithm separates the computation for~$L$
from that for~$g$.  Indeed, $g$ is formed by
the coefficients of~$L$,~$g_0$, the $\tilde g_i$ and their derivatives given in Figure~\ref{fig:HRTelescoping}.
This feature enables us to either return the certificate~$g$ as a sum of unnormalised rational functions,
or a normalised rational function.

A selection of timings by this implementation and others are given in
Table~\ref{tab:random};
our code, the full table, as well as the random inputs
are given in~\cite{OurSoft}.
For our experiments, we exhaustively
considered all 49 bidegree patterns in factorisations of denominators
$Q_1\cdots Q_m^m$ ($m\le5$) that add up to bidegree~(5,5), and
generated corresponding random denominators, imposing the
integers of the expanded forms to have around 26 digits.
Numerators were generated as random bidegree-(5,5) polynomials with
coefficients of 26 digits.

\begin{table}
\begin{scriptsize}
\tabcolsep2pt
\begin{center}
\begin{tabular}{r|rrrrrrrr}
No.& \tt AZ & \tt Abr & \tt RAZ & \tt H1 & \tt H2 & \tt HO & \tt EI & \tt MG \\
\hline
29 & 44 & 72 & 32 & 28 & 36 & 20 & 608 & 528 \\
43 & 52 & 76 & 36 & 20 & 24 & 32 & 652 & 584 \\
46 & 4268 & 1436 & 784 & 492 & 1288 & 752 & 343413 & 18945 \\
49 & 474269 & 34694 & 20977 & 10336 & 36254 & 22417 & $\infty$ & 652968
\end{tabular}
\end{center}
\end{scriptsize}
\vskip-12pt
\caption{Creative telescoping on random instances}\label{tab:random}
\begin{small}
Timings in ms for algorithms in Table~\ref{tab:algos} (stopped after 30~min).
\end{small}
\vskip-10pt
\end{table}

\smallskip\noindent {\bf Application to diagonals.}
The diagonal of a formal power series $f= \sum_{i, j\ge 0}f_{i,j}
x^i y^j$ in $k[[x, y]]$ is defined to be the power series $\Delta(f)
:= \sum_{i=0}^\infty f_{i,i} x^i$.
For a D-finite power series~$f$, it is known to be
D-finite~\cite{Lipshitz1988}, and it is even algebraic
for a bivariate rational function $f \in k(x, y) \cap
k[[x,y]]$~\cite[\S\,6.3]{Stanley1999}.
A linear differential operator $L\in k(x)\langle D_x\rangle$ that
annihilates~$\Delta(f)$ can then be computed via rational-function
telescoping, owing to the following classical lemma from~\cite{Lipshitz1988}.

\vspace{-0.2cm}
\begin{lemma}
Any telescoper for $f(y, \frac{x}{y})/y$ annihilates $\Delta(f)$.
\end{lemma}
By this lemma, it suffices to compute a telescoper without its certificate to get an annihilator.
Algorithm \textsf{HermiteTelescoping} is suitable for
this task, since it separates computation of telescopers and
certificates.
Alternatively, for $f=P/Q$, we can compute an annihilator of
$\Delta(f)$ either as the differential resolvent of the resultant
$\textrm{Res}_y(Q,P-\tau D_yQ)$, or simply \emph{guess\/} it from the
first terms of the series expansion of~$\Delta(f)$.

We compare the various algorithms
on an example borrowed from~\cite{Flaxman-2004-SMM}
(timings of execution are given in Table~\ref{tab:diags}):
\begin{equation}\label{eq:diags}
f = \frac{1}{1-x-y-xy(1-x^d)}, \ \text{where} \ d \in \bN.
\end{equation}

All computer calculations have been performed on a Quad-Core Intel
Xeon X5482 processor at 3.20GHz, with 3GB of RAM, using up to 6.5GB of
memory allocated by Maple.

\begin{table}
\begin{scriptsize}
\tabcolsep2pt
\begin{center}
\begin{tabular}{r|rrrrrrrrrr}
$d$ & \tt AZ & \tt Abr & \tt RAZ & \tt H1 & \tt H2 & \tt HO & \tt RR & \tt GHP \\
\hline
4 & 176 & 136 & 100 & 116 & 208 & 108 & 220 & 956 \\
8 & 3032 & 4244 & 4380 & 1976 & 5344 & 4396 & 10336 & 154409 \\
10 & 11740 & 12816 & 7108 & 7448 & 24565 & 7076 & 46882 & 1118313 \\
\hline
4 & 184 & 168 & 120 & 120 & 220 & 116 & 224 & 1340 \\
8 & 3540 & 3704 & 2540 & 2092 & 6976 & 2516 & 10348 & 271480 \\
10 & 16817 & 17013 & 9200 & 8068 & 32218 & 9092 & 46750 & $\infty$
\end{tabular}
\end{center}
\end{scriptsize}
\vskip-12pt
\caption{Computation of the diagonals of~\eqref{eq:diags}}\label{tab:diags}
\begin{small}
Timings in ms by creative telescoping of $f(y,x/y)/y$ (upper half) or
$f(y/x,x)/x$ (second half).  Algorithms listed in Table~\ref{tab:algos}.
\end{small}
\end{table}

\begin{table}
\vskip-3pt
\begin{scriptsize}
\begin{list}{}{\itemsep-3pt}
\item[\tt AZ] \verb+DETools[Zeilberger]+
\item[\tt Abr] \verb+AZ+ with Abramov's denominator bound by option \verb+gosper_free+
\item[\tt RAZ] Algorithm \textsf{RatAZ} of Fig.~\ref{fig:RatAZ}, with
  lower-bound prediction
\item[\tt H1] our Hermite-based approach, without certificate normalisation
\item[\tt H2] \verb+H1+, but with normalised certificate
\item[\tt HO] \textsf{RAZ}, solving~\eqref{eq:CT} by Horowitz--Ostrogradsky
\item[\tt EI] \verb+H1+ with evaluation and interpolation for calculations over~$k(x)$
\item[\tt MG] \verb+Mgfun+'s creative telescoping for general D-finite functions
\item[\tt RR] telescoper computation by resultant and differential resolvent
\item[\tt GHP] telescoper guessing by diagonal expansion and
  Hermite--Pad\'e
\end{list}
\end{scriptsize}
\vskip-15pt
\caption{List of the algorithms for the experiments}\label{tab:algos}
\vskip-10pt
\end{table}

\end{section}

{\scriptsize

\bibliographystyle{abbrv}

}

\end{document}